\documentclass[english]{ccdconf}

\usepackage{lineno,hyperref}
\modulolinenumbers[5]

\usepackage[english]{babel}
\usepackage{theorem}
\usepackage[framemethod=tikz]{mdframed}
\theoremheaderfont{\itshape\bfseries}
{\theorembodyfont{\itshape}
	\newtheorem{assumption}{\textbf{Assumption}}
	\newtheorem{lemma}{\textbf{Lemma}}
	\newtheorem{definition}{\textbf{Definition}}
	\newtheorem{theorem}{\textbf{Theorem}}
	
	\newtheorem{remark}{\textbf{Remark}}
	
	\newtheorem{problem}{\textbf{Problem}}
}
\usepackage{tabularx} 
\usepackage{graphicx}
\usepackage{epsfig}
\usepackage{newtxtext,newtxmath}
\usepackage{amsmath}
\usepackage{amsfonts}
\usepackage{mathrsfs}
\usepackage{xcolor}
\usepackage{graphicx}
\usepackage[ruled,vlined]{algorithm2e}
\usepackage{tcolorbox}
\usepackage{subfig}
\usepackage{tabulary}

\newcommand{\T}{^{\mbox{\tiny T}}}
\newcommand{\R}{\mathbb{R}}
\newcommand{\C}{\mathbb{C}}

\let\leq\leqslant
\let\geq\geqslant

\newenvironment{proof}[1][Proof]%
{\par\noindent\textit{#1:\ }}%
{\hspace*{\fill} \rule{6pt}{6pt}}
\newenvironment{proof*}[1][Proof]%
{\par\noindent\textit{#1:\ }}{}

\DeclareMathOperator{\diag}{diag}

\DeclareMathOperator{\rank}{rank}
\DeclareMathOperator{\im}{Im}

\usepackage{tikz}
\usetikzlibrary{shapes,calc,arrows,patterns,decorations.pathmorphing
	,decorations.markings}
\usetikzlibrary{arrows.meta}

\newenvironment{system}[1]%
{\setlength{\arraycolsep}{0.5mm}\left\{ \; \begin{array}{#1}}%
	{\end{array} \right.}
\newenvironment{system*}[1]%
{\setlength{\arraycolsep}{0.5mm} \begin{array}{#1}}%
	{\end{array}}

\usepackage{tikz}
\usetikzlibrary{shapes,calc,arrows,patterns,decorations.pathmorphing
	,decorations.markings}
\usetikzlibrary{arrows.meta}

\begin{document}

		\title{State Synchronization of Discrete-time Multi-agent Systems in Presence of Unknown Nonuniform Communication Delays: A Scale-free Protocol Design}

		\author{Zhenwei Liu\aref{ZW},
			Donya Nojavanzadeh\aref{WSU},
			Ali Saberi\aref{WSU},
			Anton A. Stoorvogel\aref{UT}}
		
		\affiliation[ZW]{College of Information Science and Engineering, Northeastern University, Shenyang, P.~R.~China
			\email{liuzhenwei@ise.neu.edu.cn}}
		\affiliation[WSU]{School of Electrical Engineering and Computer Science, Washington State University, Pullman,~WA,~USA
			\email{donya.nojavanzadeh@wsu.edu; saberi@wsu.edu}}
			\affiliation[UT]{Department of Electrical Engineering, Mathematics
			and Computer Science, University of Twente,~Enschede,~The Netherlands
			\email{a.a.stoorvogel@utwente.nl}}

		\maketitle
		
				\begin{abstract}
			In this paper we study scale-free state synchronization of discrete-time homogeneous multi-agent systems (MAS) subject to unknown, nonuniform and arbitrarily large communication delays. The scale-free protocol utilizes localized information exchange and is designed solely based on the knowledge of agents' model and does not require any information about the communication network and the size of the network (i.e. number of agents).
		\end{abstract}

		
		\footnotetext{This work is supported by the Nature Science Foundation of Liaoning Province, PR China under Grant 2019-MS-116 and the Fundamental Research Funds for the Central Universities of China under Grant N2004014 and the United States National Science Foundation under Grant 1635184.}
		
%
%
%
%

\section{Introduction}

Cooperative control of multi-agent systems (MAS) has become a hot topic among researchers because of its broad application in various areas such as biological systems, sensor networks, automotive vehicle control, robotic cooperation teams and so on. See for example books \cite{ren-book,wu-book,kocarev-book,bullobook}. The objective is to secure an asymptotic agreement on common states (i.e., state synchronization) or output trajectories (output synchronization) through distributed control protocols. It is worthwhile to note that state synchronization inherently requires homogeneous MAS.

In practical applications, the network dynamics are not perfect and may be
subject to delays. Time delays may afflict system performance or even
lead to instability. As discussed in
\cite{cao-yu-ren-chen}, two kinds of delays have been considered in the
literature: input delays and communication delays. Input delays encapsulate the
processing time to execute an input for each agent, whereas
communication delays can be considered as the time it takes to transmit
information from an origin agent to its destination. It is worthwhile to point out that packet drops in exchanging information can be considered as special case of communication delay, because re-sending packets after they were dropped can be easily done but just having time delay in the data transmission channels. Some researches have been done for networks subject to communication delays. Fundamentally, there are two approaches in the literature for dealing with MAS subject to communication delays.
\begin{enumerate}
	\item Standard state/output synchronization subject to regulating output to a constant trajectory.
	\item Delayed state/output synchronization.
\end{enumerate}

Both of these approaches preserves diffusiveness of the couplings (i.e. ensuring the invariance of the consensus manifold). An interesting line of research utilizing delayed synchronization formulation was introduced recently in \cite{Liu-Saberi-Stoorvogel-Li_delayed-con,Liu-Saberi-Stoorvogel-Li_delayed-dis}. These papers considered a \textit{dynamic} synchronized trajectory (i.e. any non constant synchronized trajectory). They designed protocols to achieve regulated delayed state/output synchronization in presence of communication delays where the communication graph was a directed spanning tree. On the other hand, majority of research on MAS subject to communication delay have been focused on achieving the standard output synchronization by regulating the output to \textit{constant} trajectory (see \cite{cao-yu-ren-chen,tian-liu,xiao-wang-tac,zhang-saberi-stoorvogel-ejc} and references therein). In all of the aforementioned papers, design of protocols require knowledge of the graph and size of the network. Also, we should point out that \cite{munz-papachristodoulou-allgower, munz-papachristodoulou-allgower2} give the consensus conditions for networks with higher-order but SISO dynamics. Moreover \cite{lin-jia-auto} considers second-order dynamics, but the communication delays are assumed to be known.


The main contribution of this paper is designing scale-free collaborative protocols for discrete-time homogeneous MAS subject to  communication delays such that:
\renewcommand\labelitemi{{\boldmath$\bullet$}}
\begin{itemize}
	\item State synchronization is achieved by regulating the outputs of the agents to constant trajectories. The sufficient solvability condition is provided for any arbitrary constant reference trajectory, while necessary and sufficient solvability conditions are established by restricting the constant reference trajectory to a set defined by agent models.  
	\item The scale-free protocol design is independent of information about the communication network or the size of the network.
	\item The proposed collaborative dynamic protocols can tolerate any unknown, nonuniform and arbitrarily large communication delays.  
\end{itemize}

\subsection*{Notations and preliminaries}

 We denote the set of non-negative integers by $\mathbb{Z}_{\geq 0}=\{x\in\mathbb{Z}|x\geq 0\}$. Given a matrix $A\in \mathbb{R}^{n\times m}$, $A\T$ denotes the transpose of $A$. Let $\text{\bf j}$ indicate $\sqrt{-1}$. A square matrix $A$ is said to be Schur stable if all its eigenvalues are inside the unit circle. We denote by
$\diag\{A_1,\ldots, A_N \}$, a block-diagonal matrix with
$A_1,\ldots,A_N$ as its diagonal elements.  $I_n$ denotes the
$n$-dimensional identity matrix and $0_n$ denotes $n\times n$ zero
matrix; sometimes we drop the subscript if the dimension is clear from
the context. For $\bar{A}\in \mathbb{C}^{n\times m}$ and $\bar{B}\in \mathbb{C}^{p\times q}$, the Kronecker product of $\bar{A}$ and $\bar{B}$ is defined as
\[
\bar{A}\otimes \bar{B}=\begin{pmatrix}
\bar{a}_{11}\bar{B}& \hdots& \bar{a}_{1m}\bar{B}\\
\vdots&\ddots&\vdots\\
\bar{a}_{n1}\bar{B}&\hdots&\bar{a}_{nm}\bar{B}
\end{pmatrix}
\]
where $[\bar{A}]_{ij}=\bar{a}_{ij}$. The following properties of the Kronecker product will be particularly useful.
\[
\begin{system*}{cll}
&(A\otimes B)(C\otimes D)=(AC)\otimes (BD),\\
&(\bar{A}\otimes\bar{B})\T=\bar{A}\T\otimes\bar{B}\T,\\
&\bar{A}\otimes(\bar{B}+\bar{C})=\bar{A}\otimes\bar{B}+\bar{A}\otimes\bar{C}.
\end{system*}
\]
Moreover, if $\bar{A}$ and $\bar{B}$ are nonsingular matrices, then
\[
(\bar{A}\otimes\bar{B})^{-1}=\bar{A}^{-1}\otimes\bar{B}^{-1}.
\]

To describe the information flow among the agents we associate a \emph{weighted graph} $\mathcal{G}$ to the communication network. The weighted graph $\mathcal{G}$ is defined by a triple
$(\mathcal{V}, \mathcal{E}, \mathcal{A})$ where
$\mathcal{V}=\{1,\ldots, N\}$ is a node set, $\mathcal{E}$ is a set of
pairs of nodes indicating connections among nodes, and
$\mathcal{A}=[a_{ij}]\in \mathbb{R}^{N\times N}$ is the weighted adjacency matrix with non negative elements $a_{ij}$. Each pair in $\mathcal{E}$ is called an \emph{edge}, where
$a_{ij}>0$ denotes an edge $(j,i)\in \mathcal{E}$ from node $j$ to
node $i$ with weight $a_{ij}$. Moreover, $a_{ij}=0$ if there is no
edge from node $j$ to node $i$. We assume there are no self-loops,
i.e.\ we have $a_{ii}=0$. A \emph{path} from node $i_1$ to $i_k$ is a
sequence of nodes $\{i_1,\ldots, i_k\}$ such that
$(i_j, i_{j+1})\in \mathcal{E}$ for $j=1,\ldots, k-1$. A \emph{directed tree} is a subgraph (subset
of nodes and edges) in which every node has exactly one parent node except for one node, called the \emph{root}, which has no parent node. The \emph{root set} is the set of root nodes. A \emph{directed spanning tree} is a subgraph which is
a directed tree containing all the nodes of the original graph. If a directed spanning tree exists, the root has a directed path to every other node in the tree.  

For a weighted graph $\mathcal{G}$, the matrix
$L=[\ell_{ij}]$ with
\[
\ell_{ij}=
\begin{system}{cl}
\sum_{k=1}^{N} a_{ik}, & i=j,\\
-a_{ij}, & i\neq j,
\end{system}
\]
is called the \emph{Laplacian matrix} associated with the graph
$\mathcal{G}$. The Laplacian matrix $L$ has all its eigenvalues in the
closed right half plane and at least one eigenvalue at zero associated
with right eigenvector $\textbf{1}$ \cite{royle-godsil}. Moreover, if the graph contains a directed spanning tree, the Laplacian matrix $L$ has a single eigenvalue at the origin and all other eigenvalues are located in the open right-half complex plane \cite{ren-book}.


\section{Problem Formulation}
Consider the multi-agent system composed of $N$ identical discrete-time linear agents, 
\begin{equation}\label{agent}
	\Sigma_i : \begin{system}{ccl}
		x_i(k+1) &=& Ax_i(k) +B u_i(k)\\
		y_i(k) &=& Cx_i(k)
	\end{system}
\end{equation}
where $x_i \in \R^{n}$, $y_i \in \R^{p}$, and $u_i \in \R^m$ are the state, output and the input of
agent $i\in \{1,\hdots,N\}$, respectively.

We need the following assumption.
\begin{assumption}\label{agentass2} 
	All eigenvalues of $A$ are in closed unit disc, that is agents are at most weakly unstable.
\end{assumption}
\begin{remark}
	Note that agents, satisfying assumption \ref{agentass2}, can have repeated poles on the unit circle and hence be unstable.
\end{remark}

The network provides agent $i$ with the following information
\begin{equation}\label{zeta}
	\zeta_{i}(k)=\sum_{j=1}^N a_{ij}(y_{i}(k)-y_{j}(k-\kappa_{ij})),
\end{equation}
where $\kappa_{ij} \in \mathbb{Z}_{\geq 0}$ represents an unknown communication delay from agent $j$ to agent $i$. In the above
$a_{ij}\geq 0$ and $a_{ii}=0$. This communication topology of the network, presented in \eqref{zeta}, can be associated to a weighted graph $\mathcal{G}$ with each
node indicating an agent in the network and the weight of an edge is
given by the coefficient $a_{ij}$.  The communication
delay implies that it took $\kappa_{ij}$ seconds for
agent $j$ to transfer its state information to agent $i$.

In terms of the coefficient of the associated Laplacian matrix $L$, $\zeta_i(k)$ can be
represented as
\begin{equation}\label{zeta-l1}
	\zeta_i(k)=\sum_{j=1}^N\ell_{ij}y_j(k-\kappa_{ij})
\end{equation}
where $\kappa_{ii}=0$. 
Obviously, state synchronization is achieved if 
\begin{equation}
	\lim_{k\to \infty}\left(x_i(k)-x_j(k)\right)=0 \quad \text{ for all}\quad i,j\in \{1,\hdots,N\}.
\end{equation}

Our goal is to achieve state synchronization among all agents while the synchronized output dynamic is equal to a constant reference trajectory $y_r$. We assume that a nonempty subset $\mathscr{C}$ of the agents have access to their own output relative to the reference trajectory $y_r\in\mathbb{R}^p$ . In other words, each agent has access to the quantity
\begin{equation}
	\psi_i=\iota_i(y_i- y_r), \qquad \iota_i=\begin{system}{cl}
		1, \quad i\in \mathscr{C},\\
		0, \quad i\notin \mathscr{C}.
	\end{system}
\end{equation}
Therefore, the information available for agent $i\in\{1,\ldots,N\}$, is given by
\begin{equation}\label{zeta-bar1}
	\bar{\zeta}_i(k)=\sum_{j=1}^N a_{ij}(y_{i}(k)-y_{j}(k-\kappa_{ij}))+\iota_i(y_i(k)- y_r).
\end{equation}

From now on, we will refer to the node set $\mathscr{C}$ as root set. For any graph with the Laplacian matrix $L$, we define the expanded Laplacian matrix as
\begin{equation}\label{barL}
	\bar{L}=L+\diag\{\iota_i\}=[\bar{\ell}_{ij}]_{N\times N}
\end{equation}
which is not a regular Laplacian matrix associated to the graph, since the sum of its rows need not be zero. Meanwhile, it should be emphasized that $\bar{\ell}_{ij}=\ell_{ij}$ for $i\neq j$ in $\bar{L}$. We define $\bar{D}$ as
\begin{equation}
\bar{D}=I-(2I+D_{in})^{-1}\bar{L}=[\bar{d}_{ij}]_{N\times N}
\end{equation}
where 
\[
D_{in}=\diag\{d_{in}(i)\}
\]
with $d_{in}(i)=\sum_{j=1}^{N}a_{ij}$. It is easily verified that the matrix $\bar{D}$ is a matrix with all elements nonnegative and the sum of each row is less than or equal to $1$. Then, equation \eqref{zeta-bar1} can be rewritten as
\begin{equation}\label{zeta-bar2}
	\bar{\zeta}_i(k)=\frac{1}{2+d_{in}(i)}\sum_{j=1}^N\bar{\ell}_{ij}(y_j(k-\kappa_{ij})-y_r)
\end{equation}
with $\kappa_{ii}=0$. To guarantee that each agent can achieve the required regulation, we need to make sure that there exists a pass to each node starting with node from the set $\mathscr{C}$. Therefore, we denote the following set of graphs.
\begin{definition}\label{graph-def}
	Given a node set $\mathscr{C}$, we denote by $\mathbb{G_\mathscr{C}^N}$ the set of all directed graphs with $N$ nodes containing the node set $\mathscr{C}$, such that every node of the network graph $\mathscr{G}\in \mathbb{G_\mathscr{C}^N}$ is a member of a directed tree which has its root contained in the node set $\mathscr{C}$.
	Note that this definition does not require necessarily the existence of directed spanning tree.
\end{definition}

\begin{remark}\label{Remark-Lbar}
	From \cite[Lemma 7]{grip-yang-saberi-stoorvogel-automatica} it follows for any $\mathscr{G}\in \mathbb{G_\mathscr{C}^N}$ defined in Definition \ref{graph-def}, the associated expanded Laplacian matrix $\bar{L}$ as defined by \eqref{barL} is invertible and all the eigenvalues of $\bar{L}$ have positive real parts.
\end{remark}

In this paper, we also introduce a localized information exchange among agents. In particular, each agent $i\in\{1,\hdots N\}$ has access to the following information denoted by $\hat{\zeta}_i$, of the form 
\begin{equation}\label{zeta_hat}
	\hat{\zeta}_i(k)=\frac{1}{2+d_{in}(i)}\sum_{j=1}^N \bar{\ell}_{ij}\xi_j(k-{\kappa}_{ij})
\end{equation}
where $\xi_j \in \mathbb{R}^n$ is a variable produced internally by agent $j$ and to be defined in next sections. Given that agents communicate $y_i$ and $\xi_i$ over the same communication networks, the communication delays $\kappa_{ij}$ between agent $j$ and agent $i$ are the same in equations \eqref{zeta-bar2} and \eqref{zeta_hat} with $\kappa_{ii}=0$.

We formulate the following problem of state synchronization for networks subject to unknown, nonuniform and arbitrarily large communication delays utilizing linear scale-free collaborative protocols as follows.

\begin{problem}\label{prob1}
	Consider a MAS described by \eqref{agent} and  \eqref{zeta-bar2} and a given constant reference trajectory $y_r\in\mathbb{R}^p$. Let a set of nodes $\mathscr{C}$ be given which defines the set $\mathbb{G}_\mathscr{C}^N$. 
	
	Then, the \textbf{scalable state synchronization problem based on localized information exchange utilizing collaborative protocols} for networks subject to unknown, nonuniform and arbitrarily large communication delays is to find, if possible, a linear dynamic protocol for each agent $i \in \{1,\hdots,N\}$, using only knowledge of agent model, i,e. $(A,B,C)$, of the form
	\begin{equation}\label{pro}
		\begin{system}{cl}
		{x}_{c,i}(k+1)&=A_c x_{c,i}(k)+B_{c1} \bar{\zeta}_i(k)+B_{c2} \hat{\zeta}_i(k),\\
			u_i(k)&=F_c x_{c,i}(k),
		\end{system}
	\end{equation}
where $\hat{\zeta}_i(k)$ is defined in \eqref{zeta_hat} with $\xi_i(k)=H_c x_{c,i}(k)$ and $x_{c,i}\in \mathbb{R}^{n_c}$
	such that for any $N$, any graph $\mathscr{G}\in\mathbb{G_\mathscr{C}^N}$ and any communication delays $\kappa_{ij}\in \mathbb{Z}_{\geq 0}$ we achieve	
	\begin{enumerate}
		\item[(i)] regulated output synchronization, i.e., 
		\begin{equation}\label{output_synch}
			\lim_{k\to\infty}{(y_i(k)-y_r)}=0, \quad \text{for  } i\in \{1,\hdots,N\},
		\end{equation}
		\item[(ii)] state synchronization, i.e.,
		\begin{equation}\label{state_synch}
			\lim_{k\to\infty}{(x_i(k)-x_j(k))}=0, \quad \text{for all } i,j\in \{1,\hdots,N\}.
		\end{equation}
	\end{enumerate}
\end{problem}

\section{Main Results}
Our main results are provided in the following two subsections. In the first subsection, we consider solvability of Problem \ref{prob1} for any arbitrary given constant reference trajectory $y_r\in\mathbb{R}^p$. We show that if agents are right-invertible and have no invariant zeros equal to one, Problem \ref{prob1} is solvable for any arbitrary given constant reference trajectory and we provide protocol design for this class of agents. In the second subsection, we provide necessary and sufficient conditions for solvability of Problem \ref{prob1}. We identify a set $\mathscr{Y}_r\subseteq \mathbb{R}^p$ and we show that Problem \ref{prob1} is solvable if and only if we restrict the constant reference trajectory to this set which is independent of the communication graph and obtained solely based on agent models.



\subsection{Solvability condition and protocol design for arbitrary constant reference trajectory $y_r\in \mathbb{R}^p$}\label{sec-rig-inv}

\begin{figure}[t]
	\includegraphics[width=8cm, height=5cm]{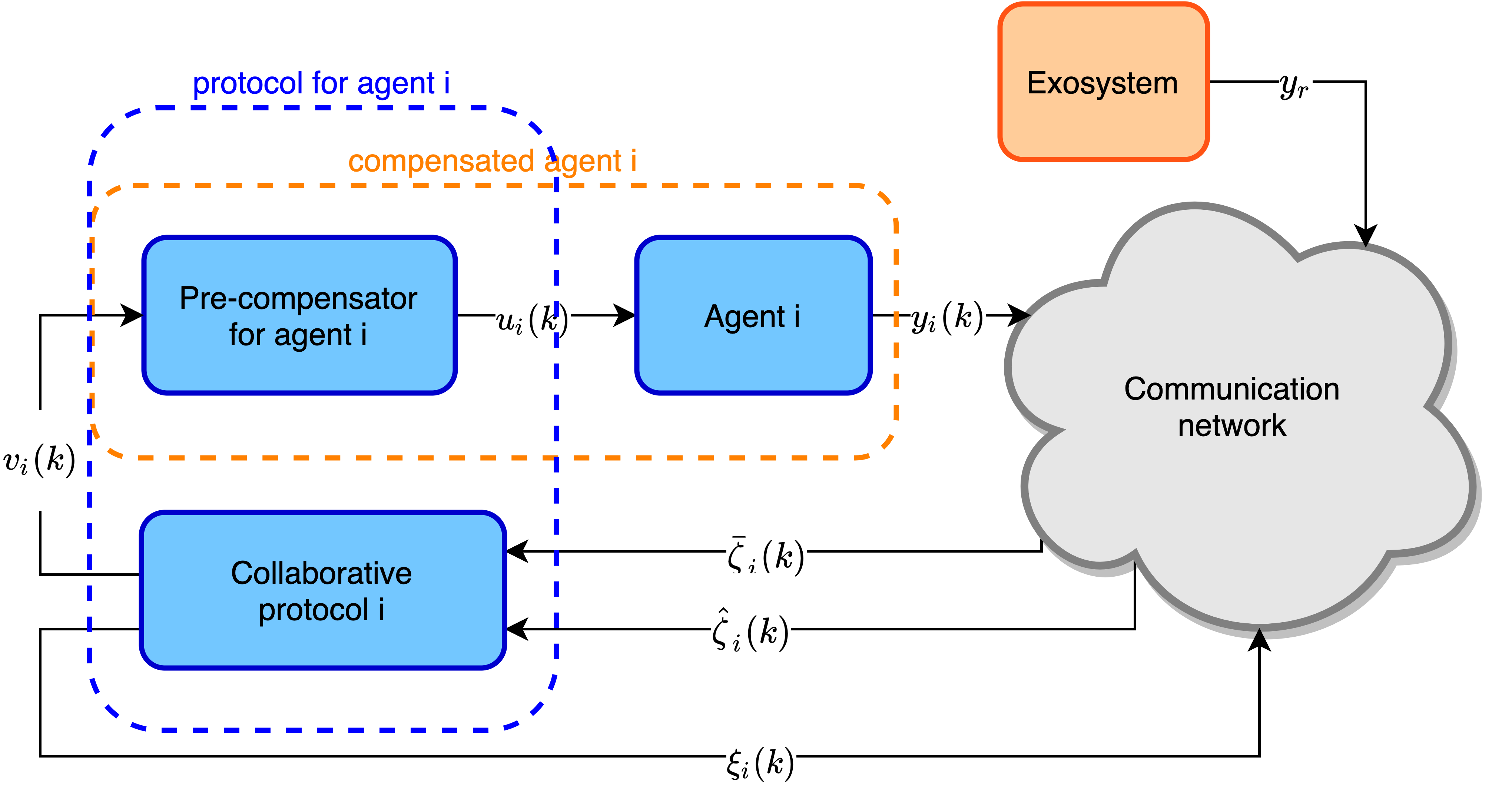}
	\centering
	\caption{Architecture of scale-free protocols}\label{communication_arch}
\end{figure}

In this subsection, we show that Problem \ref{prob1} is solvable for any given arbitrary constant trajectory $y_r\in\mathbb{R}^p$ as long as the agents are right-invertible which has no invariant zeros at one. We design protocols for this class of agents. The architecture of the protocols is shown in Figure
\ref{communication_arch}. As it is shown in the figure, the design consists of two steps. The first step is designing a pre-compensator for each agent to be able to regulate the states to a constant value. In the second step, we design collaborative protocols for the compensated agents to achieve state synchronization.

\hspace{3mm} \textbf{Step I:}  First we find an injective matrix $V$ such that 
\begin{equation}\label{regmat}
\begin{pmatrix}
A & BV\\C & 0
\end{pmatrix}
\end{equation}
is square and invertible. Such a matrix exists. To show that, we observe agent model described by $(A,B,C)$ is right-invertible and has no invariant zeros at one, hence we have the matrix
\begin{equation}\label{mat1}
\begin{pmatrix}
A & B\\C & 0
\end{pmatrix}
\end{equation}
is full-row rank. Also, due to the detectability of $(A,C)$, we have the first $n$ columns of \eqref{mat1} are linearly independent. Therefore the existence of the injective matrix $V$ is guaranteed. Next, we consider the so-called regulator equations.
\[
\begin{pmatrix}
A-I & B\\C & 0
\end{pmatrix}\begin{pmatrix}
\Pi\\ \Gamma
\end{pmatrix}=\begin{pmatrix}
0\\V
\end{pmatrix}.
\]
The invertibility of \eqref{regmat} implies that the regulator equation has a unique solution. Meanwhile, invertibility of \eqref{regmat} means that
\[
\rank{\begin{pmatrix}
	A-I & B\Gamma\\C & 0
	\end{pmatrix}}=n+\rank{\Gamma}.
\]

Then, we design the following precompensator for each agent of MAS \eqref{agent}.
\begin{tcolorbox}[colback=white]
	\begin{equation}\label{precom}
	\begin{system*}{cl}
	p_i(k+1)&=p_i(k)+\begin{pmatrix}
	0&I
	\end{pmatrix}v_i(k), \qquad\quad p_i(k)\in \mathbb{R}^v\\
	u_i(k)&=\Gamma_1 p_i(k)+\begin{pmatrix}
	\Gamma_2&0
	\end{pmatrix}v_i(k)
	\end{system*}
	\end{equation}
	where $v_i(k)$ is new input, $\Gamma_1$ is injective and satisfies $\im \Gamma=\im \Gamma_1$ with $v=\rank{\Gamma}$. Moreover, $\Gamma_2$ is chosen such that 
	\begin{equation}\label{gamam-mat}
	\begin{pmatrix}
	\Gamma_1&\Gamma_2
	\end{pmatrix}
	\end{equation}
	is square and invertible. 
\end{tcolorbox}

In order to design collaborative protocols we first obtain the compensated agents by combining \eqref{agent} and \eqref{precom} as
\begin{equation}\label{comMAS}
\begin{system*}{cl}
\bar{x}_i(k+1)&=\bar{A}\bar{x}_i(k)+\bar{B}v_i(k)\\
y_i(k)&=\bar{C}\bar{x}_i(k)
\end{system*}
\end{equation}
where
\[
\bar{x}_i(k)=\begin{pmatrix}
x_i(k)\\p_i(k)
\end{pmatrix}, \bar{A}=\begin{pmatrix}
A&B\Gamma_1\\0&I
\end{pmatrix},\bar{B}=\begin{pmatrix}
B\Gamma_2&0\\0&I
\end{pmatrix}, \bar{C}=\begin{pmatrix}
C&0
\end{pmatrix}.
\]
We also need to verify stabilizability of $(\bar{A},\bar{B})$ and detectability of $(\bar{A},\bar{C})$. The stabilizability follows immediately from the invertibility of \eqref{gamam-mat} and the stabilizability of $(A, B)$. For detectability we need to verify that 
\[
\rank{\begin{pmatrix}
	zI-A&-B\Gamma_1\\0&(z-1)I\\C&0
	\end{pmatrix}}=n+v=n+\rank{\Gamma_1}
\]
for all $z$ outside or on the unit circle. If $z\ne 1$, then it immediately follows the detectability of $(A,C)$. When $z=1$, one have 
\[
\rank{\begin{pmatrix}
	I-A&-B\Gamma_1\\0&0\\C&0
	\end{pmatrix}}=\rank{\begin{pmatrix}
	I-A&-B\Gamma_1\\C&0
	\end{pmatrix}}=n+\rank{\Gamma_1}.
\]

\hspace{3mm} \textbf{Step II:} In this step, the following linear dynamic protocol is designed for the compensated agents \eqref{comMAS} as
\begin{tcolorbox}[colback=white]
	\begin{equation}\label{pro-lin-partial}
	\begin{system}{cll}
	\hat{x}_i(k+1) &=& \bar{A}\hat{x}_i(k)-\bar{B}K\hat{\zeta}_i(k)+F(\bar{\zeta}_i(k)-\bar{C}\hat{x}_i(k))\\
	\chi_i(k+1) &=& \bar{A}\chi_i(k)+\bar{B}v_i(k)+\bar{A}\hat{x}_i(k)-\bar{A}\hat{\zeta}_{i}(k)\\
	v_i (k)&=& -K\chi_i(k),
	\end{system}
	\end{equation}
	where matrices $K$ and $F$ are such that $\bar{A}-F\bar{C}$ and $\bar{A}-\bar{B}K$ are Schur stable.
	In this protocol, agents communicate $\xi_i(k)=\chi_i(k)$, i.e. each agent has access to the localized information exchange
	\begin{equation}\label{add_1}
	\hat{\zeta}_{i}=\frac{1}{2+d_{in}(i)}\sum_{j=1}^N\bar{\ell}_{ij}\chi_j(k-{\kappa}_{ij}),
	\end{equation}

	while $\bar{\zeta}_i$ is defined via \eqref{zeta-bar2}.
\end{tcolorbox}

We formulate the following theorem.

\begin{theorem}\label{th-par-linear}
Consider a MAS described by \eqref{agent} and  \eqref{zeta-bar2} where $(A,B)$ is stabilizable and $(A,C)$ is detectable. Assume Assumption \ref{agentass2} is satisfied. Let a set of nodes $\mathscr{C}$ be given which defines the set $\mathbb{G}_\mathscr{C}^N$. 

Then, the scalable state synchronization problem utilizing localized information exchange via linear dynamic protocol as stated in Problem \ref{prob1} is solvable for any $y_r \in \mathbb{R}^p$ \textbf{if} the system represented by $(A,B,C)$ is right-invertible and has no invariant zeros at one. More specifically, under these conditions, for any given constant reference trajectory $y_r\in \mathbb{R}^p$, protocol \eqref{pro-lin-partial} and \eqref{precom} achieves scalable state synchronization for any communication delays $\kappa_{ij}\in \mathbb{Z}_{\geq 0}$ and any graph $\mathscr{G} \in \mathbb{G}_\mathscr{C}^N$ with any size of the network $N$.
\end{theorem}

To obtain the result of Theorem \ref{th-par-linear}, we need to the following lemmas.

\begin{lemma}\cite[Lemma 3]{zhang-saberi-stoorvogel-delay}\label{hode-lemma-system-c}
	Consider a linear time-delay system
	\begin{equation}\label{systeminlem}
	x(k+1)=Ax(k)+\sum_{i=1}^{m}A_{i}x(k-\kappa_i),
	\end{equation}
	where $x(k)\in\R^{n}$ and $\kappa_{i}\in \mathbb{Z}_{\geq 0}$. Assume that 
	$A+\sum_{i=1}^{m}A_{i}$ is Schur stable. Then,
	\eqref{systeminlem} is asymptotically stable, if
	\begin{equation}\label{det-cond}
	\det\left[e^{\text{\bf j}\omega} I-A- \sum_{i=1}^{m}e^{-\text{\bf j}\omega\kappa_i}A_{i}\right]\neq 0,
	\end{equation}
	for all $\omega\in[-\pi,\pi]$, and for all $\kappa_{i}\in \mathbb{Z}_{\geq 0}$ with $0<\kappa_{i}\leq\kappa^r$ for $i\in\{1,\hdots, N\}$.
\end{lemma}
\begin{lemma}\label{ljwlem} 
	Let $\beta$ be an upper bound for the eigenvalues of $\bar{D}$. Then, for all communication delays $\kappa_{ij}\in \mathbb{Z}_{\geq 0}$, $(i,j \in\{1 , \cdots , N\})$ and all $\omega\in\R$, all eigenvalues of matrix  
	\begin{multline}\label{ljw}
	\bar{D}_{\text{\bf j}\omega}(\kappa)=\\
	\begin{pmatrix}
	\bar{d}_{11}&\bar{d}_{12}e^{-\text{\bf j}\omega \kappa_{12}}&\cdots&\bar{d}_{1N}e^{-\text{\bf j}\omega \kappa_{1N}}\\
	\bar{d}_{21}e^{-\text{\bf j}\omega \kappa_{21}}&\ddots&\vdots&\bar{d}_{2N}e^{-\text{\bf j}\omega \kappa_{2N}}\\
	\vdots&\cdots&\ddots&\vdots\\
	\bar{d}_{N1}e^{-\text{\bf j}\omega \kappa_{N1}}&\cdots&\bar{d}_{N1}e^{-\text{\bf j}\omega \kappa_{NN}}&\bar{d}_{NN}
	\end{pmatrix}
	\end{multline}
	will be equal to or less than $\beta$, where $\bar{d}_{ij}$ are defined in $\bar{D}$.
\end{lemma}

\begin{proof}[Proof of Lemma \ref{ljwlem}] The proof is given in \cite[Lemma 2]{liu2018regulated}.
\end{proof}

\begin{proof}[Proof of Theorem \ref{th-par-linear}]  We need to show that protocol \eqref{pro-lin-partial} and \eqref{precom} solves Problem \ref{prob1}. First, we show that there exists a $\bar{\Pi}$ such that $\bar{A}\bar{\Pi}=\bar{\Pi}$ and $\bar{C}\bar{\Pi}=V$. Let $W$ be such that $\Gamma_1 W=\Gamma$, in that case it is easy to verify that we can choose
	\[
	\bar{\Pi}=\begin{pmatrix}
	\Pi\\W
	\end{pmatrix}.
	\]
	Let $\tilde{x}_i(k)=\bar{x}_i(k)-\bar{\Pi}y_r$, we have
	\begin{equation}
	\begin{system}{cl}
	\tilde{x}_i(k+1)&=\bar{A}{\tilde{x}}_i(k)+\bar{B}v_i(k)\\
	y_i(k)-y_r&=\bar{C}\tilde{x}_i(k)
	\end{system}
	\end{equation}
	and by defining 
	\[
	\tilde{x}(k)=\begin{pmatrix}
	\tilde{x}_1(k)\\\vdots\\\tilde{x}_N(k)
	\end{pmatrix}, \quad  
	\chi(k)=\begin{pmatrix}
	\chi_1(k)\\\vdots\\\chi_N(k)
	\end{pmatrix}
	\]
	we have the following closed-loop system in frequency domain as:	
	\begin{equation}\label{newsystem1}
	\begin{system}{cl}
	e^{\text{\bf j}\omega}{\tilde{x}}=&(I\otimes \bar{A}) \tilde{x}-(I\otimes \bar{B}K)\chi\\
	e^{\text{\bf j}\omega}{\hat{x}}=&(I\otimes \bar{A}) \hat{x}-((I-\bar{D}_{\text{\bf j}\omega}({\kappa}))\otimes \bar{B}K)\chi\\
	&\qquad+((I-\bar{D}_{\text{\bf j}\omega}({\kappa})) \otimes F\bar{C})\tilde{x}-(I \otimes F\bar{C})\hat{x}\\
	e^{\text{\bf j}\omega}{\chi}=&(I\otimes (\bar{A}-\bar{B}K)) \chi-((I-\bar{D}_{\text{\bf j}\omega}({\kappa})) \otimes \bar{A})\chi+(I\otimes \bar{A})\hat{x}
	\end{system}
	\end{equation}
	where $\bar{D}_{\text{\bf j}\omega}(\kappa)$ is defined by \eqref{ljw}. Let $\delta=\tilde{x}-\chi$, and $\bar{\delta}=((I-\bar{D}_{\text{\bf j}\omega}({\kappa}) \otimes I) \tilde{x}-\hat{x}$. Then, we obtain  
	\begin{equation}\label{newsystem2}
	\begin{system}{cll}
    e^{\text{\bf j}\omega}{\tilde{x}}=&(I\otimes (\bar{A}-\bar{B}K)) \tilde{x}+(I\otimes \bar{B}K)\delta\\
	e^{\text{\bf j}\omega}{\delta}=&(\bar{D}_{\text{\bf j}\omega}(\kappa)\otimes \bar{A})\delta+(I\otimes \bar{A})\bar{\delta}\\
	e^{\text{\bf j}\omega}{\bar{\delta}}=&\left(I\otimes(\bar{A}-F\bar{C})\right)\bar{\delta}
	\end{system}
	\end{equation}
	We need to show the asymptotic stability of \eqref{newsystem2} for all communication delays $\kappa_{ij}\in \mathbb{Z}_{\geq 0}$. Since $\bar{A}-F\bar{C}$ is stable, then we have $\bar{\delta}\to 0$ as $k \to \infty$. As such asymptotic stability of \eqref{newsystem2} is implied by asymptotic stability of the following reduced system.
	\begin{equation}\label{newsystemfreq2}
	\begin{pmatrix}
	e^{\text{\bf j}\omega}{\tilde{x}}\\e^{\text{\bf j}\omega}{\delta}
	\end{pmatrix}=
	\begin{pmatrix}
	I\otimes (\bar{A}-\bar{B}K) & I\otimes  \bar{B}K \\
	0 &\bar{D}_{\text{\bf j}\omega}(\kappa)\otimes \bar{A}
	\end{pmatrix}\begin{pmatrix}
	\bar{x}\\\delta
	\end{pmatrix}
	\end{equation}
	Following Lemma \ref{hode-lemma-system-c}, we prove the stability of \eqref{newsystemfreq2} in two steps. In the first step, we prove the stability in the absence of communication delays and in the second step we prove the stability of \eqref{newsystemfreq2} by checking condition \eqref{det-cond}.
	\begin{enumerate}
		\item When there is no communication delay in the network, the stability of system \eqref{newsystem2} is equivalent to asymptotic stability of the matrix
		\begin{equation}
		\begin{pmatrix}
		I\otimes (\bar{A}-\bar{B}K) & I\otimes  \bar{B}K \\
		0 & \bar{D}\otimes \bar{A}
		\end{pmatrix}
		\end{equation}		
			where $\bar{D}=[\bar{d}_{ij}]\in\R^{N\times N}$ and we have that the eigenvalues of $\bar{D}$  are in open unit disk.
		The eigenvalues of $\bar{D}\otimes A$ are of the form
		$\lambda_i \mu_j$, with $\lambda_i$ and $\mu_j$ eigenvalues of
		$\bar{D}$ and $A$, respectively \cite[Theorem 4.2.12]{horn-johnson}. Since $|\lambda_i|<1$ and
		$|\mu_j|\leq 1$, we find $\bar{D}\otimes A$ is Schur stable. Then we have
		\begin{equation}\label{estable}
		\lim_{k\to \infty}\delta_i(k)\to 0.
		\end{equation}
		Therefore, we have that the dynamics for $\delta_i(k)$ is asymptotically stable.
		Then, we just need to prove the stability of
		\begin{equation}\label{statefeedback3}
		\tilde{x}(k+1)=(I\otimes (\bar{A}-\bar{B}K)) \tilde{x}(k)
		\end{equation}
		which $\bar{A}-\bar{B}K$ is Schur stable. Therefore, we can obtain the asymptotic stability of \eqref{newsystem2}, i.e.,
		\[
		\lim_{k\to \infty}\tilde{x}_i(k)\to 0.
		\]
		It implies that $x_i(k)-\Pi y_r\to0$, i.e. $x_i(k)\to x_j(k)$.
		
		\item Next, in the light of Lemma \ref{hode-lemma-system-c}, the closed-loop system \eqref{newsystemfreq2} is asymptotically stable for all communication delays $\kappa_{ij} \in \mathbb{Z}_{\geq 0}$, if		
		\begin{equation}\label{cond}
		\det \begin{bmatrix}
		e^{\text{\bf j}\omega} I-\begin{pmatrix}
		I\otimes (\bar{A}-\bar{B}K) & I\otimes  \bar{B}K \\
		0 & \bar{D}_{\text{\bf j}\omega}(\kappa)\otimes \bar{A}
		\end{pmatrix}
		\end{bmatrix}\ne 0
		\end{equation}
		for all $\omega \in \mathbb{R}$ and any communication delays $\kappa_{ij} \in \mathbb{R}_{\geq 0}$. Inequality \eqref{cond} is satisfied if the matrix 
		\begin{equation}\label{clmatrixfreq}
		\begin{pmatrix}
		I\otimes (\bar{A}-\bar{B}K) & I\otimes  \bar{B}K \\
		0 & \bar{D}_{\text{\bf j}\omega}(\kappa)\otimes \bar{A}
		\end{pmatrix}
		\end{equation}
		does not have any eigenvalue on the unit circle for all $\omega \in [-\pi,\pi]$ and any communication delays $\kappa_{ij} \in \mathbb{Z}_{\geq 0}$.
		In the light of Lemma \ref{ljwlem}, we have that all eigenvalues of $\bar{D}_{\text{\bf j}\omega}({\kappa})$ are in open unit disc for any ${\kappa}_{ij}$. Therefore 
		\[
		\bar{D}_{\text{\bf j}\omega}(\kappa)\otimes \bar{A}
		\]
		has all eigenvalues in open unit disc. It implies that all eigenvalues of matrix \eqref{clmatrixfreq} are in open unit disc, i.e. matrix \eqref{clmatrixfreq} does not have any eigenvalue on the unit circle for all $\omega \in [-\pi,\pi]$ and any communication delays $\kappa_{ij} \in \mathbb{Z}_{\geq 0}$. Thus we have
		\[
		\tilde{x}_i(k)\to 0 \text{ i.e. } x_i(k)\to \Pi y_r
		\]
		which means the synchronization $x_i(k)\to x_j(k)$ is achieved.
	\end{enumerate}
\end{proof}

\subsection{Necessary and sufficient solvability conditions and protocol design for general agent model}\label{set1}

In this subsection, we provide necessary and sufficient conditions for solvability of Problem \ref{prob1}. We define a set $\mathscr{Y}_r\subseteq \mathbb{R}^p$ solely based on agent models and we show that problem \ref{prob1} is solvable if and only if the the constant reference trajectory belongs to this set. In this case, plant can be general and non right-invertible. We begin first by defining set $\mathscr{Y}_r$ as following.
\begin{equation*}
	\begin{system*}{cl}
		\mathscr{Y}_r&=\bigg\{y \in \mathbb{R}^p\bigg|\begin{pmatrix}
			0\\y\end{pmatrix} \in \im \begin{pmatrix}
			A-I&B\\C&0\end{pmatrix}\bigg\}\\
		&=\{y \in \mathbb{R}^p| \exists x \in \mathbb{R}^n , u \in \mathbb{R}^m :Ax+Bu=x , Cx=y\}.
	\end{system*}
\end{equation*}
 Note that $\mathscr{Y}_r=\mathbb{R}^p$ if $(A,B,C)$ is right-invertible and without invariant zeros at one.

Next, for a given $y_r\in \mathscr{Y}_r$, we provide protocol design which has the same architecture as previous subsection. The first step is designing a pre-compensator for each agent and the second step is designing collaborative protocols for the compensated agents to achieve state synchronization. 

\hspace{3mm} \textbf{Step I:} 
Let $R$ be an injective matrix such that $\mathscr{Y}_r=\im R$. In this case, we can find the matrices $\Pi$ and $\Gamma$ such that:	
\begin{equation}\label{reg2}
\begin{pmatrix}
0\\R
\end{pmatrix}=\begin{pmatrix}
A-I&B\\C&0
\end{pmatrix}\begin{pmatrix}
\Pi\\\Gamma
\end{pmatrix}
\end{equation}
and 
\begin{equation}\label{condrank}
\rank{\begin{pmatrix}
	A-I&B\Gamma\\C&0
	\end{pmatrix}}=n+\rank{\Gamma}.
\end{equation}
Given that $(A,C)$ detectable, the first $n$ columns of $\begin{pmatrix}
A-I&B\Gamma\\C&0
\end{pmatrix}$ are linearly independent. If \eqref{condrank} is not satisfied, then there exist $x$ and $v$ such that 
\[
\begin{pmatrix}
A-I&B\Gamma\\C&0
\end{pmatrix}\begin{pmatrix}
x\\v
\end{pmatrix}=0
\]
with $B\Gamma v\ne0$ and $v\T v=1$. On the other hand, we have
\[
\begin{pmatrix}
A-I&B\\C&0
\end{pmatrix}\begin{pmatrix}
\Pi-xv\T\\\Gamma(I-vv\T)
\end{pmatrix}=0
\]
It shows that $\tilde{\Pi}=\Pi-xv\T$ and $\tilde{\Gamma}=\Gamma(I-vv\T)$ also satisfy the above equation but with $\rank{\tilde{\Gamma}}<\rank{\Gamma}$. Recursively, we can find a solution of \eqref{reg2} such that the rank condition \eqref{condrank} is satisfied.

Then, similar to precompensator design in Subsection \ref{sec-rig-inv}, with $\Gamma$ obtained as above, we have the following precompensator.
\begin{tcolorbox}[colback=white]
	\begin{equation}\label{precom2}
		\begin{system*}{cl}
	p_i(k+1)&=p_i(k)+\begin{pmatrix}
	0&I
	\end{pmatrix}v_i(k), \qquad\quad p_i(k)\in \mathbb{R}^v\\
	u_i(k)&=\Gamma_1 p_i(k)+\begin{pmatrix}
	\Gamma_2&0
	\end{pmatrix}v_i(k)
	\end{system*}
	\end{equation}
	where $v_i(k)$ is new input, $\Gamma_1$ is injective and satisfies $\im \Gamma=\im \Gamma_1$ with $v=\rank{\Gamma}$. Moreover, $\Gamma_2$ is chosen such that 
	\begin{equation}\label{gamam-mat2}
	\begin{pmatrix}
	\Gamma_1& \Gamma_2
	\end{pmatrix}
	\end{equation}
	is square and invertible. 
\end{tcolorbox}

In order to design collaborative protocols we first obtain the compensated agents by combining \eqref{agent} and \eqref{precom2} as

\begin{equation*}\label{comMAS2}
\begin{system*}{cl}
\bar{x}_i(k+1)&=\bar{A}\bar{x}_i(k)+\bar{B}v_i(k)\\
y_i(k)&=\bar{C}\bar{x}_i(k)
\end{system*}
\end{equation*}
	where 
	\[
	\bar{x}_i(k)=\begin{pmatrix}
	x_i(k)\\p_i(k)
	\end{pmatrix}, \bar{A}=\begin{pmatrix}
	A&B\Gamma_1\\0&I
	\end{pmatrix},\bar{B}=\begin{pmatrix}
	B\Gamma_2&0\\0&I
	\end{pmatrix}, \bar{C}=\begin{pmatrix}
	C&0
	\end{pmatrix}.
	\]
We need to verify the stabilizability and detectability of the compensated system. The stability follows immediately from \eqref{gamam-mat2} and the stabilizability of $(A, B)$, and for detectability we need to verify that 
	\[
	\rank{\begin{pmatrix}
		zI-A&-B\Gamma_1\\0&(z-1)I\\C&0
		\end{pmatrix}}=n+v
	\]
	where $v$ is such that $\Gamma_1 \in \mathbb{R}^{n \times v}$ for all $z$ outside or on the unit circle. For $z\ne 1$, this immediately follows form the detectability of $(C,A)$. For $z=1$, we have 
	\begin{multline*}
	\rank{\begin{pmatrix}
		I-A&-B\Gamma_1\\0&0\\C&0
		\end{pmatrix}}=\rank{\begin{pmatrix}
		I-A&-B\Gamma\\C&0
		\end{pmatrix}}\\=n+\rank{\Gamma_1}=n+v.
	\end{multline*}
Since $\rank{\Gamma}=\rank{\Gamma_1}$ and $\rank{\Gamma_1}=v$ (since $\Gamma_1$ is injective), we can obtain $(\bar{C},\bar{A})$ is detectable.	

\hspace{3mm} \textbf{Step II:} In this step, we design collaborative protocol for the compensated agents similar to collaborative protocol designed in Subsection \ref{sec-rig-inv}. 	
\begin{tcolorbox}[colback=white]
\begin{equation}\label{pro-lin-partial2}
\begin{system}{cll}
\hat{x}_i(k+1) &=& \bar{A}\hat{x}_i(k)-\bar{B}K\hat{\zeta}_i(k)+F(\bar{\zeta}_i(k)-\bar{C}\hat{x}_i(k))\\
\chi_i(k+1) &=& \bar{A}\chi_i(k)+\bar{B}v_i(k)+\bar{A}\hat{x}_i(k)-\bar{A}\hat{\zeta}_{i}(k)\\
v_i (k)&=& -K\chi_i(k),
\end{system}
	\end{equation}
	where matrices $K$ and $F$ are such that $\bar{A}-F\bar{C}$ and $\bar{A}-\bar{B}K$ are Schur stable.
	In this protocol, agents communicate $\xi_i(k)=\chi_i(k)$, i.e. each agent has access to the localized information exchange
\begin{equation}\label{add_2}
	\hat{\zeta}_{i}=\frac{1}{2+d_{in}(i)}\sum_{j=1}^N\bar{\ell}_{ij}\chi_j(k-{\kappa}_{ij}),
\end{equation}

while $\bar{\zeta}_i$ is defined via \eqref{zeta-bar2}.
\end{tcolorbox}

Then, we have the following theorem.
	\begin{theorem}\label{th-par-linear2}
	Consider a MAS described by \eqref{agent} and  \eqref{zeta-bar2} where $(A,B)$ is stabilizable and $(A,C)$ is detectable. Assume Assumption \ref{agentass2} is satisfied. Let a set of nodes $\mathscr{C}$ be given which defines the set $\mathbb{G}_\mathscr{C}^N$. 
	
	Then, the scalable state synchronization problem with localized information exchange via linear dynamic protocol as stated in Problem \ref{prob1} is solvable \textbf{if and only if} $y_r \in \mathscr{Y}_r$. More specifically, for any $y_r \in \mathscr{Y}_r$, protocol \eqref{pro-lin-partial2} and \eqref{precom2} achieves scalable state synchronization for any communication delays $\kappa_{ij}\in \mathbb{Z}_{\geq 0}$ and any graph $\mathscr{G} \in \mathbb{G}_\mathscr{C}^N$ with any size of the network $N$.
\end{theorem}

\begin{proof}[Proof of Theorem \ref{th-par-linear2}]
	\begin{enumerate}
			\item \emph{Necessity:} In order agents track a constant reference trajectory signal $y_r$, there must exists $x_0$ and $u_0$ such that 	
	\begin{equation}
	\begin{pmatrix}
	A-I&B\\C&0
	\end{pmatrix}\begin{pmatrix}
	x_0\\u_0
	\end{pmatrix}=\begin{pmatrix}
	0\\y_r
	\end{pmatrix}
	\end{equation}	
	Clearly, such $x_0$ and $u_0$ exist only if $y_r$ belongs to the set $\mathscr{Y}_r$, that proves the necessary condition.
	
%

%
%
		\item \emph{Sufficiency:} 
		For the sufficiency, we need to show that protocol \eqref{pro-lin-partial2} and \eqref{precom2} solves Problem \ref{prob1}. The proof is exactly the same as proof of Theorem \ref{th-par-linear} except for the choice of $\bar{\Pi}$ and $\tilde{x}$. In this case, we choose $\tilde{x}_i=\bar{x}_i-\bar{\Pi}z$ where $z$ is such that $y_r=Rz$. Moreover, we set
			\[
			\bar{\Pi}=\begin{pmatrix}
			\Pi\\ W
			\end{pmatrix}
			\]
			where $W$ is such that $\Gamma_1 W=\Gamma$. It is then easily seen that $\bar{A}\bar{\Pi}=\bar{\Pi}$ and $\bar{C}\bar{\Pi}=R$.
	\end{enumerate}
\end{proof}

\section{Numerical Examples}

The aim of this section is to show the scalability and effectiveness of our protocol design via numerical examples. To show the scalability, we consider three networks with different communication graphs, different number of agents. We will show that we achieve scale-free state synchronization with our one-shot designed protocol. We also illustrate that our protocol can tolerate arbitrarily large communication delays.

Consider the agents model \eqref{agent} as
\begin{equation*}\label{ex1}
	\begin{cases}
		{x}_i(k+1)=\begin{pmatrix}
			-1&0&0\\0&\frac{1}{2}&\frac{\sqrt{3}}{2}\\0&-\frac{\sqrt{3}}{2}&\frac{1}{2}&
		\end{pmatrix}x_i(k)+\begin{pmatrix}
			1&0\\0&1\\0&0
		\end{pmatrix}u_i(k),\\
		y_i(k)=\begin{pmatrix}
			1&0&1
		\end{pmatrix}x_i(k)
	\end{cases}
\end{equation*}
\begin{figure}[t]
	\includegraphics[width=9cm, height=10cm]{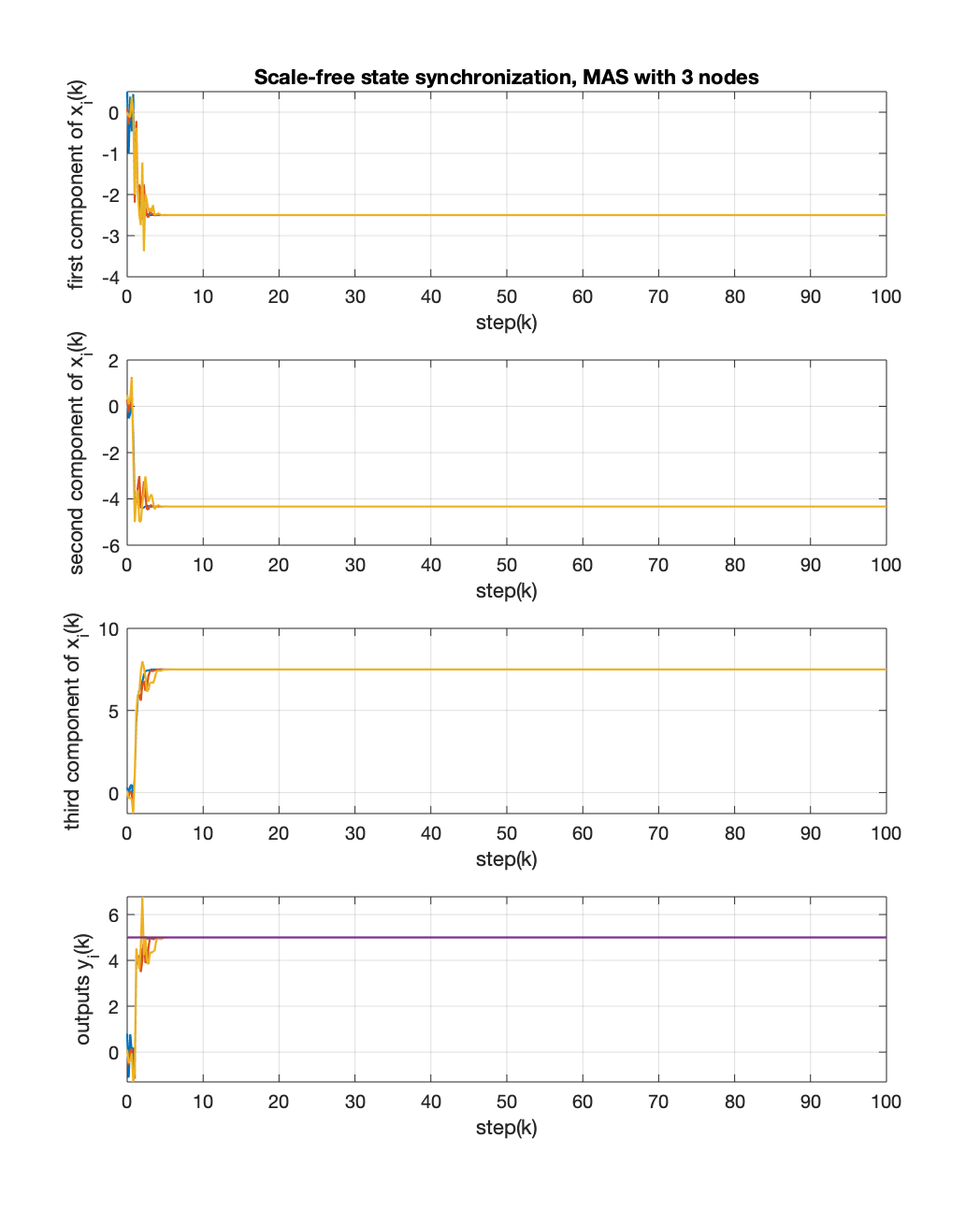}
	\centering
	\caption{State synchronization of discrete-time MAS with $N=3$ agents and unknown nonuniform communication delays}\label{3Nodes}
\end{figure}
\begin{figure}[t]
	\includegraphics[width=9cm, height=10cm]{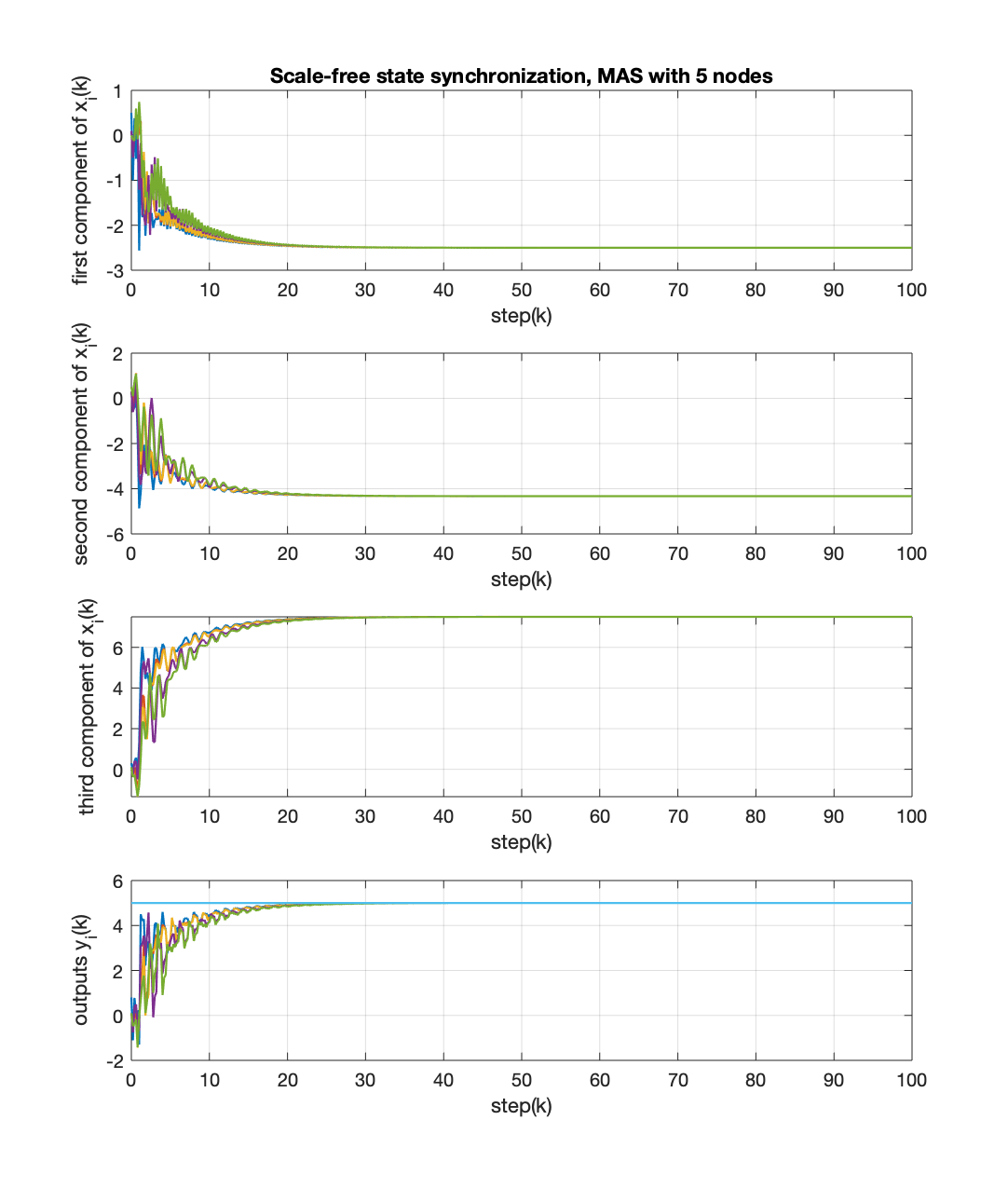}
	\centering
	\caption{State synchronization of discrete-time MAS with $N=5$ agents and unknown nonuniform communication delays}\label{5Nodes}
\end{figure}
\begin{figure}[t]
	\includegraphics[width=9cm, height=10cm]{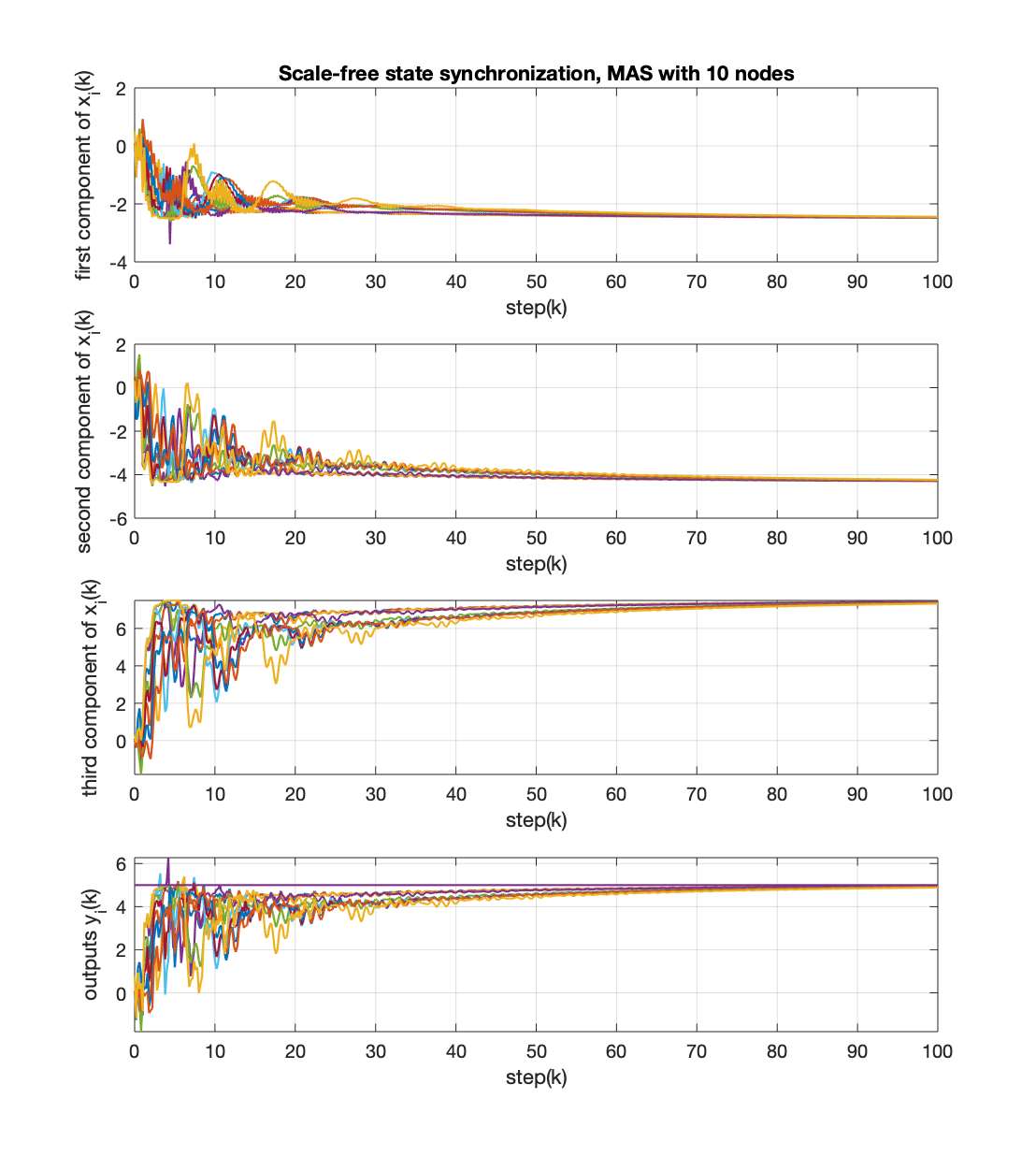}
	\centering
	\caption{State synchronization of discrete-time MAS with $N=10$ agents and unknown nonuniform communication delays}\label{10Nodes}
\end{figure}
By choosing
\[
R=1, \Pi=\begin{pmatrix}
-\frac{1}{2}&-\frac{\sqrt{3}}{2}&\frac{3}{2}
\end{pmatrix}\T, \Gamma=\Gamma_1=-\begin{pmatrix}
1\\\sqrt{3}
\end{pmatrix}, \Gamma_2=\begin{pmatrix}
0\\1
\end{pmatrix}.
\]
We design our pre-compensators as 
	\begin{equation*}
\begin{system}{cl}
p_i(k+1)&=p_i(k)+\begin{pmatrix}
0&1
\end{pmatrix}v_i(k),\\
u_i(k)&=-\begin{pmatrix}
1\\\sqrt{3}
\end{pmatrix}p_i(k)+\begin{pmatrix}
0&0\\1&0
\end{pmatrix}v_i(k)
\end{system}
\end{equation*}
We also choose matrix $K$ and $F$ as following such that $\bar{A}-\bar{B}K$ and $\bar{A}-F\bar{C}$ are Schur stable.
\[
K=\begin{pmatrix}
    0.54  &  0.87  &  0.62 &  -1.12\\
-0.89 &  -0.35  &  0.15   & 0.12
\end{pmatrix}, F=\begin{pmatrix}
   -0.45\\
-0.19\\
1.05\\
0.34
\end{pmatrix}
\]

 Then, our one-shot-designed protocol for the compensated agents would be as following:
	\begin{equation}\label{ex}
\begin{system}{cl}
{\hat{x}}_i(k+1) =& \begin{pmatrix}
-0.54 &        0 &   0.45  & -1\\
0.19   &0.5  &  1.05  & -1.73\\
-1.05  & -0.86 & -0.55  &      0\\
-0.34     &    0  & -0.34   & 1
\end{pmatrix}\hat{x}_i(k)\\
&-\begin{pmatrix}
         0    &     0    &     0     &    0\\
0.54 &   0.87 &  0.62  & -1.12\\
0     &    0     &    0     &    0\\
-0.89 &  -0.35   & 0.15 &   0.12\\
\end{pmatrix}\hat{\zeta}_i(k)+\begin{pmatrix}
   -0.45\\
-0.19\\
1.05\\
0.34
\end{pmatrix}\bar{\zeta}_i(k),\\
{\chi}_i(k+1) =& \begin{pmatrix}
   -1      &   0    &     0  & -1\\
-0.54  &-0.37  &  0.24  & -0.61\\
0  & -0.86  &  0.5    &    0\\
0.89  & 0.35 &  -0.15  &  0.87\\
\end{pmatrix}\chi_i(k)\\
&+\begin{pmatrix}
-1 &        0    &     0  & -1\\
0   & 0.5  &  0.86 &   -1.73\\
0   & -0.86 &   0.5   &     0\\
0    &     0   &      0  &  1\\
\end{pmatrix}(\hat{x}_i(k)-\hat{\zeta}_{i}(k)), \\
v_i (k)=& -\begin{pmatrix}
    0.54  &  0.87 &   0.62  & -1.12\\
-0.89  & -0.35  &  0.15  & 0.12
\end{pmatrix}\chi_i(k).
\end{system}
\end{equation}

 In all the following three examples we choose $y_r=5$.

\begin{enumerate}
	\item  Firstly, we consider a MAS with $3$ agents, $N =3$ and communication network with associated adjacency matrix $\mathcal{A}_1$, where $a_{21}=a_{32}=1$. Communication delays are chosen as $\kappa_{21}=0.5, \kappa_{21}=0.8$. The results of scale-free state synchronization are presented in Figure \ref{3Nodes}.
	\item Next, we consider a MAS with $5$ agents, $N=5$, and communication network with associated adjacency matrix $\mathcal{A}_2$, where $a_{13}=a_{21}=a_{25}=a_{32}=a_{35}=a_{43}=a_{54}=1$. Communication delays are chosen as $\kappa_{15}=0.5$, $\kappa_{43}=1$ and the rest are equal to zero. The simulation results are shown in Figure \ref{5Nodes}.
	
	\item Finally, we consider a MAS with $10$ agents and communication network with associated adjacency matrix $\mathcal{A}_3$, where $a_{21}=a_{5,10}=a_{32}=a_{43}=a_{54}=a_{65}=a_{76}=a_{87}=a_{98}=a_{10,9}=a_{15}=1$. Communication delays are chosen as $\kappa_{32}=1$, $\kappa_{43}=3, \kappa_{65}=2, \kappa_{10,4}=5$ and the rest are equal to zero. The simulation results are shown in Figure \ref{10Nodes}.
\end{enumerate}
	

\section{Conclusion}
In this paper we have proposed scale-free protocol design utilizing localized information exchange for state synchronization of homogeneous discrete-time MAS subject to unknown, nonuniform and arbitrarily large communication delays. The necessary and sufficient solvability conditions also has been provided. It should be emphasized that the proposed protocols were designed solely based on the knowledge of the agent models without any information about the communication networks such as bounds on the spectrum of the Laplacian matrix associated to the communication graph and the size of the network.

\bibliographystyle{plain}
\bibliography{referenc}

\end{document}